\documentclass[twocolumn, aps,preprintnumbers,showpacs]{revtex4-1}

\usepackage{latexsym}
\usepackage{amsmath}
\usepackage{amssymb}
\usepackage{amsfonts}
\usepackage{ifthen}
\usepackage{revsymb}
\usepackage{yfonts}
\usepackage{graphicx}

\usepackage{graphicx}
\usepackage{amsmath}
\usepackage{amssymb}
\usepackage{amsthm}

\newcommand{\be}{\begin{equation}}
\newcommand{\ee}{\end{equation}}
\newcommand{\ba}{\begin{array}}
\newcommand{\ea}{\end{array}}
\newcommand{\bea}{\begin{eqnarray}}
\newcommand{\eea}{\end{eqnarray}}

\newcommand{\calB}{{\cal B }}
\newcommand{\calC}{{\cal C }}

\newcommand{\calL}{{\cal L }}

\newcommand{\calP}{{\cal P }}

\newcommand{\calS}{{\cal S }}
\newcommand{\calG}{{\cal G }}

\newcommand{\calQ}{{\cal Q }}

\newcommand{\CC}{\mathbb{C}}

\newcommand{\FF}{\mathbb{F}}
\newcommand{\la}{\langle}
\newcommand{\ra}{\rangle}

\newcommand{\nn}{\nonumber}
\newcommand{\rank}{\mathop{\mathrm{rank}}\nolimits}

\newcommand{\bare}{\mathrm{bare}}
\newcommand{\Crow}{{C_{\mathrm{row}}}}
\newcommand{\Ccol}{{C_{\mathrm{col}}}}
\newcommand{\Grow}{{G_{\mathrm{row}}}}
\newcommand{\Gcol}{{G_{\mathrm{col}}}}
\newcommand{\drow}{{d_{\mathrm{row}}}}
\newcommand{\dcol}{{d_{\mathrm{col}}}}

\newtheorem{dfn}{Definition}
\newtheorem{lemma}{Lemma}
\newtheorem{prop}{Proposition}
\newtheorem{theorem}{Theorem}

\newtheorem{fact}{Fact}

\begin{document}

\title{Subsystem codes with spatially local generators}

\author{Sergey \surname{Bravyi}}
\affiliation{IBM T.J. Watson Research
Center,  Yorktown Heights  NY 10598, USA}

\date{\today}

\begin{abstract}
We study subsystem codes whose gauge group has local generators in the 2D geometry.
It is shown that there exists a family of such codes defined on lattices of size $L\times L$
with the number of logical qubits $k$ and the minimum distance $d$ both proportional to $L$.
 The gauge group
of these codes involves only two-qubit generators of type $XX$ and $ZZ$ coupling nearest neighbor qubits
(and some auxiliary one-qubit generators). Our proof is not constructive as it relies on
a certain version of the Gilbert-Varshamov bound for classical codes. Along the way we introduce and study
properties of generalized Bacon-Shor codes which might be of independent interest.
Secondly, we prove that any   2D subsystem
$[n,k,d]$ code with spatially local generators obeys upper bounds
$kd=O(n)$ and $d^2=O(n)$. The analogous upper bound proved recently for  2D stabilizer codes
is $kd^2=O(n)$. Our results thus demonstrate that subsystem codes can be more powerful than stabilizer codes
under the spatial locality constraint.
\end{abstract}

\maketitle

\section{Introduction}

Fault-tolerant quantum information processing based on 2D topological quantum codes has received a considerable
attention lately since it  can be implemented on  quantum machines with a
geometrically local architecture.
The potential of topological codes as a viable alternative to concatenated quantum codes
was first realized by Dennis et al~\cite{Dennis01}. It was shown in~\cite{Dennis01} that an active error correction
in the 2D toric code permits reliable storage of a logical qubit if the error rate
in the quantum hardware
 is below the threshold
value about $1\%$. The threshold for storage of a qubit
has been recently improved by Andrist et al~\cite{Andrist10}
by using  topological color codes~\cite{BMD:topo,KBAM:unionj} instead of the toric code. The success of topological codes was extended from a storage of a qubit to the
universal quantum computation by making use of the powerful technique known as
code deformations~\cite{Raussendorf06,Raussendorf07,Bombin09+,DiVin09}.
These recent developments demonstrate that topological codes provide an attractive framework for design of new fault-tolerant protocols.

In order to better understand the potential of topological codes for storing and manipulating quantum information,
it is desirable to derive fundamental bounds on the parameters of quantum codes that stem from the spatial locality constraint
and find families of codes that achieve these bounds.
A progress in this direction has been recently made for 2D stabilizer codes~\cite{BT08,KC08,Bravyi09,Yoshida10}.
Such codes can be defined in a system of  $n$ physical qubits occupying sites of a regular
square lattice of size $\sqrt{n}\times \sqrt{n}$.
Quantum information is encoded into a codespace $\calL$ spanned
by common eigenvectors of pairwise commuting $n$-qubit Pauli operators $S_1,\ldots,S_m$
known as {\em stabilizers}, that is,
\[
\calL=\{ |\psi\ra \in (\CC^2)^{\otimes n}\, : \,S_a\, |\psi\ra =|\psi\ra \quad \mbox{for all $a$}\}.
\]
The locality condition is imposed by demanding that each stabilizer $S_a$ acts non-trivially only on a constant number of qubits
located within constant distance from each other. A code has $k$ logical qubits if $\dim{\calL}=2^k$. It was shown in~\cite{Bravyi09} that any 2D stabilizer code obeys a bound
\be
\label{bound1}
kd^2=O(n),
\ee
where $d$
is the minimum distance of a code, i.e., the minimum weight of a Pauli operator
commuting with all stabilizers and implementing a non-trivial transformation on $\calL$.
The bound Eq.~(\ref{bound1}) is tight in the sense that for any given $k$ and $d$
one can construct a 2D stabilizer $[n,k,d]$ code with $n=O(kd^2)$,
see~\cite{Bravyi09} for details.

The main conclusion of the present paper is that the bound Eq.~(\ref{bound1}) can
be violated in a dramatic way for 2D {\em subsystem} codes.
Recall that a subsystem code~\cite{Bacon06,Poulin06} can be regarded as an ordinary stabilizer
code in which only part of the logical qubits is used to store information.
Accordingly, the codespace of subsystem codes can be decomposed as
\[
\calL=\calL_{\mathrm{logical}}\otimes \calL_{\mathrm{gauge}}, \quad \dim{\calL_{\mathrm{logical}}}=2^k,
\]
where  $\calL_{\mathrm{logical}}$ is the logical subsystem used to store quantum information, while
$\calL_{\mathrm{gauge}}$ represents the unused logical qubits
usually called gauge qubits.
The distance $d$ of a subsystem code is the minimum weight of a Pauli operator
commuting with all stabilizers and acting non-trivially on the logical subsystem $\calL_{\mathrm{logical}}$.
The presence of the unused gauge qubits
provides much more flexibility in the design of fault-tolerant gates and the error correction
for subsystem codes
since one does not need to worry how a particular computational operation or an error affects the gauge qubits,
see for instance~\cite{Aliferis07,Cross07}.
By the same token,  one should expect that spatial locality constraints lead to  less severe
restrictions on the parameters of subsystem codes.

One can characterize a subsystem code by its {\em gauge group} $\calG$
that includes all stabilizers and all logical operators on the unused logical qubits,
see Section~\ref{sec:stabilizer} for more details.
We shall study 2D subsystem codes for which the gauge group has spatially local generators,
that is, $\calG=\la G_1,\ldots,G_m\ra$ where
each generator $G_a$ acts non-trivially only on a constant number of qubits
located within constant distance from each other (as for stabilizers $S_a$, they may or may not
be spatially local).
For such codes the tradeoff between $n$ and $d$ was characterized in~\cite{BT08}
by showing that
\be
\label{old_bound}
d=O(\sqrt{n}).
\ee
This bound is tight since the 2D Bacon-Shor code~\cite{Bacon06} has parameters $d=\sqrt{n}$ and $k=1$.
The question that remained open is whether 2D subsystem codes may have
a better scaling of $k$ compared with 2D stabilizer codes.

In the present paper we answer this question in positive  by proving that there exist a family of 2D subsystem $[n,k,d]$
codes with both $k$ and $d$ proportional to $\sqrt{n}$.
More precisely, we prove the following.
\begin{theorem}
\label{thm:main}
Let $\alpha>0$ and $0<\beta<1/2$ be any constants such that
$\alpha + H_2(\beta)<1$, where $H_2(\beta)$ is the binary entropy.
Then for all sufficiently large integers $m$ and for all $k\le \alpha m$ there exists
a 2D subsystem $[2m^2,k,d]$ code for some $d\ge \beta m$. The gauge group
of this code has two-qubit generators of type $XX$ and $ZZ$ coupling
nearest-neighbor qubits and some one-qubit generators.
\end{theorem}
In contrast, all previously known 2D subsystem codes with the distance proportional to
$\sqrt{n}$,
such as the 2D Bacon-Shor code~\cite{Bacon06} or the topological subsystem codes~\cite{Bombin09}
encode only $k=O(1)$ qubits.

The proof  of Theorem~\ref{thm:main} relies on generalized
Bacon-Shor codes that we  introduce in Section~\ref{sec:GBS}. One can define a generalized Bacon-Shor code
for any binary matrix $A$ by placing physical qubits at the cells of $A$ for which $A_{i,j}=1$
and leaving the remaining cells empty. The gauge group $\calG$ has generators of two types:
each pair of qubits $c,c'$ located in the same row of $A$ contributes a generator $X_c X_{c'}$
and each pair of qubits $c,c'$ located in the same column of $A$ contributes a generator $Z_c Z_{c'}$.
We show that the code $\calG$ has parameters $[n,k,d]$, where $n$ is the number of non-zero matrix elements in $A$,
$k$ is the binary rank of $A$, while $d$ is determined by the minimum distance
of the classical codes spanned by columns and rows of $A$,
see Theorem~\ref{thm:GBS} in Section~\ref{sec:GBS}. Then we employ the Gilbert-Varshamov bound to prove existence of binary matrices
with the desired properties, see Theorem~\ref{thm:GV} in Section~\ref{sec:GV}. Finally, we show how to transform any generalized
Bacon-Shor code into the spatially local form by introducing ancillary qubits and simulating each long-range generator
by a chain of nearest-neighbor couplings, see Section~\ref{sec:local}.

Our second result is a new upper bound on the parameters of 2D subsystem codes
whose gauge group has spatially local generators, namely,
\be
\label{bound2}
kd=O(n).
\ee
It can be regarded as a generalization of Eq.~(\ref{bound1}) to subsystem codes.
This bound is tight up to a constant factor since one can achieve a scaling $k\sim d\sim \sqrt{n}$,
see Theorem~\ref{thm:main}.
We also prove that the original bound Eq.~(\ref{bound1}) holds for any 2D subsystem code
in which both stabilizer group and the gauge group have spatially local generators.
The topological subsystem codes of~\cite{Bombin09,Bombin10} provide an example of such codes.

The proof of the bound Eq.~(\ref{bound2}) presented in Sections~\ref{sec:tech},\ref{sec:holo},\ref{sec:proof}
requires some heavy machinery that builds upon
techniques developed in~\cite{BT08,Bravyi09,Yoshida10}.
Our first tool is  the identity relating the number of logical operators supported in two complementary
regions of a lattice, see Lemma~\ref{lemma:cleaning} in Section~\ref{sec:tech}. It was originally proved by Yoshida and Chuang~\cite{Yoshida10} for stabilizer codes. In the present paper we generalize this identity to subsystem codes
using techniques of~\cite{BT08}.
Our second tool is what we call a holographic principle for error correction, see Section~\ref{sec:holo}.
It asserts that a non-trivial logical operator cannot be supported in a region whose {\em perimeter} is smaller than the
distance of the code. The analogous result was proved in~\cite{Bravyi09} for 2D stabilizer codes
although the proof given in~\cite{Bravyi09} cannot be generalized to subsystem codes.
In Section~\ref{sec:proof} we combine these technical tools to prove the upper bound Eq.~(\ref{bound2}).

In Section~\ref{sec:concl} we summarize our results and discuss some open problems such as possible extensions
of our constructions to 3D subsystem codes.

\section{Stabilizer and subsystem codes}
\label{sec:stabilizer}
The purpose of this section is to summarize the necessary facts pertaining to stabilizer
and subsystem codes.
The main idea of stabilizer codes is to encode $k$ logical qubits into $n$ physical qubits using a
codespace $\calL\subseteq (\CC^2)^{\otimes n}$ spanned by
states $|\psi\ra$ that are invariant under the action of a {\it stabilizer group} $\calS$,
\[
\calL=\{|\psi\ra\in (\CC^2)^{\otimes n}\, : \, P\, |\psi\ra=
|\psi\ra\quad \forall P\in \calS\}.
\]
All stabilizers $P\in \calS$ must be  Pauli operators, that is, $n$-fold tensor products of the
single-qubit Pauli operators
\[
X=\left[ \ba{cc} 0 & 1 \\ 1 & 0 \\ \ea \right],
\quad
Y=\left[ \ba{cc} 0 & -i \\ i & 0 \\ \ea \right],
\quad
Z= \left[ \ba{cc} 1 & 0 \\ 0 & -1 \\ \ea \right],
\]
and the identity operators $I$.
Such tensor products generate the Pauli group on $n$ qubits,
\[
\calP=\{ \gamma\, P_1\otimes P_2\otimes \cdots \otimes P_n\},
\]
where $P_a\in \{I,X,Y,Z\}$ and $\gamma\in \{1,-1,i,-i\}$ is an overall phase factor that we shall often ignore.
Non-trivial stabilizer codes correspond to Abelian stabilizer groups $\calS\subset \calP$
such that $-I\notin \calS$.

Logical operators of a stabilizer code $\calS$ are Pauli operators preserving the codespace $\calL$.
Equivalently, logical operators are Pauli operators commuting with every element of $\calS$.
Such operators generate the centralizer of $\calS$ in the Pauli group,
\[
\calC(\calS)=\{ P\in \calP\, : \, PQ=QP \quad \forall \; Q\in \calS\}.
\]
One can always decompose the centralizer  as
\[
\calC(\calS)=\la \calS, \overline{X}_1,\overline{Z}_1,\ldots,\overline{X}_k,\overline{Z}_k\ra,
\]
where $\overline{X}_i,\overline{Z}_i$ are the logical Pauli operators, while elements of $\calS$
correspond to the logical identity operator. Here and below we use the notation $\la \ldots \ra$
for a subgroup generated by a family of operators.

For any Pauli operator $P\in \calP$ let $\mathrm{Supp}(P)\subseteq \Lambda$
be the {\em support} of $P$, that is, the subset of qubits on which $P$ acts non-trivially (by $X$, $Y$, or $Z$).
We shall use the notation $|P|=|\mathrm{Supp}(P)|$ for the weight of $P$, that is, the number of qubits in its support.

The minimum distance $d$ of a stabilizer code $\calS$ is defined as the minimum weight
of a non-trivial logical operator, that is,
\[
d=\min_{P\in \calC(\calS)\backslash \calS}\; |P|.
\]
We shall use a notation $[n,k,d]$ for a stabilizer code encoding $k$ logical qubits
into $n$ physical qubits with the distance $d$.

To define a subsystem  $[n,k,d]$ code it is convenient to start from a  stabilizer
code  $\calS$ with $k+g$ logical qubits for some $g>0$. Let $\overline{X}_i,\overline{Z}_i$, $i=1,\ldots,k$ be the logical
Pauli operators on the first $k$ logical qubits that will be used to encode information. The remaining unused logical operators
$\overline{X}_i,\overline{Z}_i$, $i=k+1,\ldots,k+g$, together with stabilizers $\calS$ generate the {\em gauge group}~\cite{Poulin06}
\[
\calG=\la \calS,\overline{X}_{k+1},\overline{Z}_{k+1},\ldots,\overline{X}_{k+g},\overline{Z}_{k+g}\ra.
\]
We will assume that the $n$ physical qubits live at vertices of a regular
2D lattice $\Lambda$ of size $\sqrt{n}\times \sqrt{n}$ with open or
periodic boundary conditions.
Given this 2D geometry we demand that  the gauge group $\calG$ must have spatially local generators,
that is, $\calG=\la G_1,\ldots,G_m\ra$, where the support of any generator $G_m$
can be bounded by a square box of size $r\times r$ for some constant interaction range $r$.
Given the gauge group $\calG$, one can compute the stabilizer group $\calS$ using the identity
\[
\calS=\calG\cap \calC(\calG),
\]
where $\calC(\calG)$ is the centralizer of $\calG$ in the Pauli group.
Note that $\calS$ may or may not have spatially local generators.
For example, any stabilizer of the 2D Bacon-Shor code~\cite{Bacon06} has weight
at least $\sqrt{n}$. On the other hand, for the topological subsystem codes~\cite{Bombin09}
generators of $\calS$ have geometry of closed loops of constant size. In general,
one can think of generators of $\calG$ as stabilizers broken into ``local chunks" such that an eigenvalue
of any stabilizer can be inferred by measuring eigenvalue of sufficiently many generators of $\calG$.
Subsystem codes with an Abelian gauge group are equivalent to ordinary stabilizer codes (such codes
have no gauge qubits and thus $\calG=\calS$).

Logical operators of a subsystem codes are Pauli operators preserving the codespace $\calL$.
Equivalently, logical operators are elements of the centralizer
\[
\calC(\calS)=\la \calG, \overline{X}_1,\overline{Z}_1,\ldots,\overline{X}_k,\overline{Z}_k\ra.
\]
Since the encoding is defined only modulo gauge operators, non-trivial logical operators
are elements of $\calC(\calS)$ that are not in $\calG$.
In the case of subsystem codes we shall often use the term {\em bare logical operators}
which refers to elements of $\calC(\calG)\backslash \calG$.
Bare logical operators preserve the codespace $\calL$ and act trivially on the gauge qubits.
They should not be confused with
{\em dressed logical operators}  which are elements of $\calC(\calS)\backslash \calG$. The identity
\[
\calC(\calS)=\calC(\calG)\cdot \calG
\]
implies that any dressed logical operator can be represented as a product of a bare logical operator
and a gauge operator.
The minimum distance of a subsystem code is defined as the minimum weight of a non-trivial dressed logical operator,
\[
d=\min_{P\in \calC(\calS)\backslash \calG}\, |P| = \min_{P\in \calC(\calG)\backslash \calG}\; \min_{G\in \calG}
\; |PG|.
\]

\section{Generalized Bacon-Shor codes}
\label{sec:GBS}

In this section we introduce a generalization of the 2D Bacon-Shor code~\cite{Bacon06}
and describe its main properties. For the sake of clarity we shall ignore the issue of
spatial locality until Section~\ref{sec:local}.

Let $A$ be an arbitrary matrix of size $m\times m$ with entries $0$ and $1$.
We shall label cells of $A$ by pairs of indices $c=(i,j)$.
Let $n=|A|$ be the  Hamming weight of $A$, i.e., the total number of non-zero matrix elements.
To define a subsystem code associated with $A$ let us place a physical qubit at each
cell $c=(i,j)$ with $A_{i,j}=1$. Hence there are totally $n$ physical qubits.
The remaining cells for which $A_{i,j}=0$ are kept solely for illustrative purposes
since there are no qubits located at these cells.
A subsystem code associated with $A$ has a gauge group
$\calG$ generated  according to the following rules:
\begin{itemize}
\item Every pair of qubits $c,c'$ located in the same row of $A$ contributes a
generator $X_c X_{c'}$
\item Every pair of qubits $c,c'$ located in the same column of $A$ contributes a
generator $Z_c Z_{c'}$
\end{itemize}
In the special case when $A_{i,j}=1$ for all $(i,j)$ the code $\calG$ coincides with the
standard 2D Bacon-Shor code~\cite{Bacon06}.
Consider as an example a binary matrix
\[
A=\left[ \ba{ccc}
1 & 1 & 0 \\
0 & 1 & 1 \\
1 & 0 & 1\\
\ea
\right].
\]
The corresponding subsystem code has $n=6$ physical qubits. The corresponding gauge group
$\calG$ has generators
\bea
\calG&=&\la X_{1,1} X_{1,2}, \; X_{2,2} X_{2,3}, \; X_{3,1} X_{3,3}, \nn \\
&& Z_{1,1} Z_{3,1},\; Z_{1,2} Z_{2,2},\; Z_{2,3} Z_{3,3}\ra. \nn
\eea
Let us explain how to  relate the parameters of the code $\calG$ to certain algebraic properties of the matrix $A$.
It will be convenient to introduce a linear subspace $\Ccol\subseteq \{0,1\}^m$ spanned
by the columns of $A$ and a linear subspace $\Crow\subseteq \{0,1\}^m$ spanned by the rows of $A$
(throughout this section all linear subspaces are defined over the binary field $\FF_2$).
One can regard $\Ccol$ and $\Crow$ as classical codes
encoding $k=\rank{(A)}$ bits into $m$ bits.
Here $\rank{(A)}$ is the rank of $A$ over the binary field $\FF_2$.
 In the example of the 2D Bacon-Shor code the matrix $A$ has rank $1$ and
 $\Ccol$, $\Crow$ are one-dimensional subspaces that include only the all-zeros and the all-ones vectors.
 Let $\dcol$ and $\drow$ be the minimum Hamming weight of non-zero vectors
in $\Ccol$ and $\Crow$ respectively. In other words,
$\dcol$ and $\drow$ is the minimum distance of the classical code $\Ccol$ and $\Crow$ respectively.
\begin{theorem}
\label{thm:GBS}
Let $A$ be an arbitrary binary matrix and $\calG$ be the subsystem
code associated with $A$ as described above. Then $\calG$ encodes
$k=\rank{(A)}$ qubits into $n=|A|$ qubits with the minimum distance
$d=\min{(\drow,\dcol)}$.
\end{theorem}

Before we proceed with the proof of the theorem
let us make several remarks concerning the definition of $\calG$.
Firstly,  the set of generators of $\calG$ introduced above might be overcomplete.
For example, suppose $c,c',c''$ is a triple of consecutive qubits that belong to the same row of $A$.
Then clearly $X_c X_{c''}=(X_c X_{c'})(X_{c'} X_{c''})$, that is, it suffices to retain
only the generators $X_c X_{c'}$ and $X_{c'} X_{c''}$.
In general, it suffices to retain  generators $X_c X_{c'}$ and $Z_c Z_{c'}$ for consecutive pairs of qubits $c,c'$ that
belong to the same row and the same column of $A$ respectively.
 Secondly, the generators of $\calG$ are not necessarily spatially local.
For example, if some row of $A$ contains an isolated `1' separated by a long string of zeros on the
left and on the right, there is no choice of spatially local generators for $\calG$.
We shall explain how to circumvent with difficulty in Section~\ref{sec:local} by placing two ancillary qubits into each
cell of $A$ with $A_{i,j}=0$ and splitting each long-range generator into a chain of short-range generators.
Finally, let us point out that a similar but technically different construction of subsystem codes
was described by Bacon and Casaccino in~\cite{BC06}. The construction of~\cite{BC06}
starts from a pair of classical linear codes $C_1=[n_1,k_1,d_1]$ and $C_2=[n_2,k_2,d_2]$.
A quantum subsystem code is then defined by placing a physical qubit at {\em every} cell of a
matrix $A$ of size $n_1\times n_2$. The $X$-part of the gauge group is defined by replicating the parity checks of $C_1$  in every column of $A$ (in the $X$-basis). Similarly, the $Z$-part of the gauge group
is defined by replicating the parity checks of $C_2$  in every row of $A$ (in the $Z$-basis).
The resulting subsystem code has parameters $[n_1n_2,k_1k_2,\min{(d_1,d_2)}]$.
The main difference between our construction and the one of~\cite{BC06}  is that we allow to use different classical codes
in different rows and columns of $A$ (although each of these codes is simply the repetition code on some subset of
qubits). Also, as one can see from Theorem~\ref{thm:GBS}, the two constructions result in quantum codes
with different parameters.

Using  Theorem~\ref{thm:GBS} one can easily get upper bounds on the number of logical qubits $k$ and
the distance $d$
of generalized Bacon-Shor codes. Firstly, we claim that for any binary matrix $A$
one has  $\drow \dcol\le n$.
Indeed, since any column of $A$ has  weight at least $\dcol$,
the matrix $A$ must contain at least $\dcol$ non-zero rows. Each  of these rows must have
weight at least $\drow$. Hence the number of non-zero matrix elements in $A$
is at least $\drow \dcol$. Applying Theorem~\ref{thm:GBS} one arrives at
\be
\label{GBSbound1}
d^2\le \drow \dcol\le  n.
\ee
Although it is not necessary, let us point out that the bound Eq.~(\ref{GBSbound1}) is not as good as it could be.
In Appendix~A we shall prove a slightly stronger bound
\be
\label{GBSbound1+}
2\drow \dcol (1-2^{-k})\le n
\ee
and construct a family of binary matrices that achieves this bound.

Furthermore, since  $A$ must contain at least $k$ non-zero rows and at least $k$ non-zero columns,
we conclude that $k\drow\le n$ and $k\dcol\le n$. Theorem~\ref{thm:GBS} then implies that
\be
\label{GBSbound2}
kd\le n.
\ee
In Section~\ref{sec:GV} we shall use Gilbert-Varshamov bound to prove that there exists
a family of codes that achieves the bounds Eqs.~(\ref{GBSbound1},\ref{GBSbound2})
up to a constant factor
asymptotically in the limit $n\to \infty$. The corresponding binary $m\times m$ matrices
$A$ are defined for all sufficiently large $m$ and obey the scaling $\rank{(A)}\ge \alpha m$,
and $\dcol,\drow\ge \beta m$ for some constants $\alpha,\beta>0$.
It leads to the scaling $k\ge \alpha m$, $d\ge \beta m$, and $\beta^2 m^2 \le n\le m^2$.

\begin{proof}[\bf Proof of Theorem~\ref{thm:GBS}]
We shall use notations $X_{i,j}$ and $Z_{i,j}$ for the Pauli operators
acting on a qubit located at a cell $(i,j)$.
Let us begin by describing the centralizer $\calC(\calG)$. For any row $i$ define a row operator
$R_i$ acting by $Z$ on every qubit located in the $i$-th row:
\[
R_i=\prod_{j\, : \, A_{i,j}=1} Z_{i,j}.
\]
Similarly, for any column $j$ define a column operator $C_j$
acting by $X$ on every qubit located in the $j$-th column:
\[
C_j=\prod_{i\, :\, A_{i,j}=1} X_{i,j}.
\]
\begin{prop}
\label{prop:C(G)}
The centralizer $\calC(\calG)$ is generated by the row and column operators,
\be
\label{C(G)}
\calC(\calG)=\la R_1,\ldots,R_m,C_1,\ldots,C_m\ra.
\ee
\end{prop}
\begin{proof}
Indeed, let us check that $R_i\in \calC(\calG)$ for any row $i$.
A generator  $X_c X_{c'}$  located at the row $i$
anti-commutes with $R_i$  at both
cells $c$ and $c'$. Thus  $X_c X_{c'}$  commutes with $R_i$. In addition, $R_i$ commutes
with generators $X_c X_{c'}$ located at rows $i'\ne i$ since their supports do not overlap.
It also commutes with generators $Z_c Z_{c'}$ since they are both operators of $Z$-type.
 Hence $R_i\in \calC(\calG)$.
The same reasoning shows that  $C_j\in \calC(\calG)$.
Conversely, let $P^Z\in \calC(\calG)$
be a Pauli operator of $Z$-type.
If  $P^Z$ acts by $Z$ on some qubit $c$, it must act by $Z$ on every other qubit $c'$ located in the same row
since $P^Z$  has to commute with all generators $X_c X_{c'}$ in this row.
It shows that $P^Z$ is a product of the row operators over some subset of rows. Similarly, any operator
$P^X\in \calC(\calG)$  of $X$-type is a product of the column operators over some subset of columns.
It proves Eq.~(\ref{C(G)}).
\end{proof}

Note that the supports of $R_i$ and $C_j$ overlap on exactly one qubit if $A_{i,j}=1$
and do not overlap if $A_{i,j}=0$. It follows that the matrix $A$
controls the commutation rules between the row and column operators, namely,
\be
\label{RC}
R_i C_j =(-1)^{A_{i,j}}\, C_j R_i
\ee
for all pairs $i,j$.
Using Proposition~\ref{prop:C(G)} we can parameterize  any $X$-type operator $P^X\in \calC(\calG)$
by a binary string $x\in \{0,1\}^m$ such that
\be
\label{Px}
P^X=\prod_{j=1}^m C_j^{x_j}=\prod_{i,j\, : \, A_{i,j} x_j=1}\; X_{i,j}.
\ee
Similarly, any $Z$-type operator $P^Z\in \calC(\calG)$
can be parameterized by a binary string $z\in \{0,1\}^m$ such that
\be
\label{Pz}
P^Z=\prod_{i=1}^m R_i^{z_i}=\prod_{i,j\, : \, z_i A_{i,j}=1}\; Z_{i,j}.
\ee
The commutation rules Eq.~(\ref{RC}) then imply that
\be
\label{crule}
P^X P^Z =(-1)^{z^T A x}\, P^Z P^X, \quad z^T A x\equiv \sum_{i,j} A_{i,j} z_i x_j.
\ee
Recall that the stabilizer group $\calS$ of a subsystem code is defined
as $\calS=\calG\cap \calC(\calG)$. From Eq.~(\ref{crule}) we infer
that $P^X$ commutes with all elements of $\calC(\calG)$ iff $x\in \mathrm{Ker}(A)$.
In this case one has $P^X\in \calC(\calC(\calG))=\calG$, that is, $P^X\in \calS$.
The same argument shows that $P^Z\in \calS$ iff $z\in \mathrm{Ker}(A^T)$.
Thus stabilizers of $X$-type and $Z$-type can be identified with
right and left zero-vectors of $A$ respectively.
Using the standard Gram-Schmidt orthogonalization
one can choose $k=\rank{(A)}$ pairs of operators
$P^X_a,P^Z_a\in \calC(\calG)$, $a=1,\ldots,k$, such that
$P^X_a$ are linear combinations of the column operators,
$P^Z_a$ are linear combinations of the row operators,
\[
P^X_a P^Z_b =(-1)^{\delta_{a,b}}\, P^Z_b P^X_a,
\]
and
\[
\calC(\calG)=\la \calS, P^X_1,\ldots,P^X_k,P^Z_1,\ldots,P^Z_k\ra.
\]
It shows that $P^X_a,P^Z_a$ are the bare logical Pauli operators
and the code $\calG$ has $k$ logical qubits.

Let us now determine the distance of $\calG$.
Consider some bare logical operator $P^X$ defined in Eq.~(\ref{Px}).
Let us analyze how one can
reduce  the weight of $P^X$ by multiplying it with the gauge operators.
Obviously, if $P^X$ has even weight in some row $i$,
that is, $\sum_j A_{i,j} x_j=0 \pmod{2}$,
one can completely cancel $P^X$ in this
row by multiplying it with the generators $X_cX_{c'}$ located in this row.
On the other hand, if $P^X$ has odd weight in some row $i$,
that is, $\sum_j A_{i,j} x_j=1 \pmod{2}$, the best one can do is to reduce
the weight of $P^X$ in this row down to $1$. Also it is clear that multiplying $P^X$ with
gauge operators of $Z$-type cannot decrease its weight. It shows that the
minimum weight of a dressed logical operator corresponding to $P^X$
is equal to the number of rows $i$ for which  $\sum_j A_{i,j} x_j=1 \pmod{2}$.
The number of such rows is nothing but the Hamming weight of the vector $Ax$.
 Note that $Ax\ne 0$  whenever $P^X$ is a non-trivial logical operator since $Ax=0$
 implies that $P^X\in \calS$, see above.
We conclude that the minimum weight of $X$-type dressed operators coincides
with the minimum Hamming weight of a vector $Ax$ where $x\notin \mathrm{Ker}(A)$.
Since such vectors span the subspace $\Ccol$, their minimum weight coincides with the distance
$d_{col}$. Similar arguments show that the minimum weight of $Z$-type dressed operators coincides
with $d_{row}$.
\end{proof}

\section{Gilbert-Varshamov bound for binary matrices}
\label{sec:GV}
Let $\calB(m,k)$ be the set of  all $m\times m$ binary matrices with the  rank $k$ over the binary field $\FF_2$.
Given any matrix $A\in \calB(m,k)$ let $\Ccol(A)\subseteq \{0,1\}^m$
be the linear subspace spanned by the columns of $A$.
Similarly, let $\Crow(A)\subseteq \{0,1\}^m$
be the linear subspace spanned by the rows of $A$.
Let $\dcol(A)$ and $\drow(A)$ be the minimum
distance of the classical code $\Ccol(A)$ and $\Crow(A)$ respectively.
We shall use the notation $H_2(p)=-p\log_2{p}-(1-p)\log_2{(1-p)}$ for the binary entropy.
The main result of this section is the following theorem.
\begin{theorem}[\bf Gilbert-Varshamov bound]
\label{thm:GV}
Let $\alpha>0$
and $0<\beta<1/2$ be any constants such that $\alpha<1-H_2(\beta)$.
Then for all sufficiently large $m$ and for all $k\le \alpha m$ there exists  a matrix $A\in \calB(m,k)$
such that $\dcol(A)\ge \beta m$ and $\drow(A)\ge \beta m$.
\end{theorem}
\begin{proof}
Our arguments will rely Gilbert-Varshamov bound for classical codes as presented in~\cite{CS95}.

For any non-zero vector $x\in \{0,1\}^m$ let $W(x)$ be the number of matrices
$A\in \calB(m,k)$ such that $x\in \Ccol(A)$.
We claim that $W(x)$ does not depend on $x$ as long as $x$ is a non-zero vector.
Indeed, let $e_1=(10\ldots0)$ be the first basis vector of $\{0,1\}^m$.
Let $W$ be the number of matrices $A\in \calB(m,k)$ such that $e_1\in \Ccol(A)$.
Any non-zero $x\in \{0,1\}^m$ can be written as $x=Re_1$ for some $R\in \calB(m,m)$.
By definition of the column space we have $x\in \Ccol(A)$ iff $x=Ay$ for
some $y\in \{0,1\}^m$. Hence
\[
x\in \Ccol(A) \quad \mbox{iff} \quad e_1\in \Ccol(R^{-1} A).
\]
Since the map $A\to RA$ defines a bijection from the set $\calB(m,k)$ to itself,
we conclude that $W(x)=W$. Since the rank of a matrix is invariant under transpositions,
we conclude that for any non-zero $x\in \{0,1\}^m$ the number of matrices
$A\in \calB(m,k)$ such that $x\in \Crow(A)$ also equals $W$.

Define a matrix $\Gamma_{x,A}$ labeled by non-zero $x\in \{0,1\}^m$ and $A\in \calB(m,k)$
such that
\[
\Gamma_{x,A}=\left\{ \ba{rcl}
1 &\mbox{if} & x\in \Ccol(A), \\
0 && \mbox{otherwise}. \\
\ea
\right.
\]
For any $A\in \calB(m,k)$ one has $\sum_{x\ne 0} \Gamma_{x,A} =2^k-1$ since
$\Ccol(A)$ has dimension $k$. Also for any $x\ne 0$ one has
$\sum_{A\in \calB(m,k)}\,  \Gamma_{x,A} = W$.
Hence we arrive at
\be
\label{W}
(2^k-1) |\calB(m,k)| = (2^m-1) W.
\ee
Let $d=\lceil\beta m\rceil$ and
$N$ be the number of matrices $A\in \calB(m,k)$ such that
$\Ccol(A)$ or $\Crow(A)$ contains a non-zero vector with weight smaller than $d$.
It suffices to check that for the chosen $d$ and any $k\le \alpha m$ one has $N<|\calB(m,k)|$.
Applying the union bound one gets
\[
N\le 2W\sum_{i=1}^{d-1} {m \choose i}\le 2W\cdot  2^{m H_2((d-1)/m)}\le 2W 2^{mH_2(\beta)}.
\]
Here we used the assumption $\beta\le 1/2$.
Using Eq.~(\ref{W}) and a bound $(2^m-1)^{-1}\le 2^{1-m}$  we arrive at
\[
N \le 4\cdot 2^{k-m+mH_2(\beta)}\cdot |\calB(m,k)|.
\]
We conclude that $N<|\calB(m,k)|$ whenever
\[
m(\alpha-1+H_2(\beta))<-2.
\]
Since we assumed that $\alpha<1-H_2(\beta)$, it holds for all sufficiently large $m$.
\end{proof}

\section{Breaking up the long-range generators}
\label{sec:local}
Let $A$ be an arbitrary $m\times m$ binary matrix and $\calG$ be the subsystem $[n,k,d]$ code
associated with $A$.
As was pointed out in Section~\ref{sec:GBS}, one can choose an independent set of generators of $\calG$
that includes operators  $X_c X_{c'}$ for consecutive pairs of qubits $c,c'$ that
belong to the same row, and operators $Z_c Z_{c'}$ for consecutive pairs of qubits $c,c'$ that
belong to the same column.
We shall consider the 2D geometry in which cells of the matrix $A$ are identified with
sites of a two-dimensional lattice of size $m\times m$.

Consider a pair of consecutive qubits $c$, $c'$ in some row $i$ and the corresponding
generator $X_c X_{c'}$. In general the cells $c$ and $c'$ are not nearest neighbors, so there might be one or several empty cells
 in the $i$-th row (those with $A_{i,j}=0$) between $c$ and $c'$.
Let these empty cells be $c_1,\ldots, c_p$, see Fig.~\ref{fig:cells}.
\begin{figure}[htb]
\centerline{
\mbox{
 \includegraphics[height=1cm]{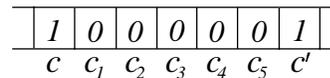}
 }}
\caption{The cells $c$ and $c'$ represent a pair of consecutive qubits in some row of $A$
that are not nearest neighbors. The cells $c_1,\ldots,c_p$ are empty.}
 \label{fig:cells}
\end{figure}

Our strategy will be to simulate the long-range generator $X_cX_{c'}$ by adding an ancillary qubit
at every cell $c_1,\ldots, c_p$ and connecting $c$, $c'$ by a chain of
short-range generators,
\be
\label{chain}
X_c X_{c'}=(X_c X_{c_1})\cdot (X_{c_1} X_{c_2})\cdots (X_{c_p} X_{c'}).
\ee
More formally, we shall define a new subsystem code $\calG'$ obtained from $\calG$
by taking out the long-range generator $X_c X_{c'}$ and adding all the short-range generators
that appear in the right-hand side of Eq.~(\ref{chain}). In addition,  $\calG'$
will contain one-qubit generators $Z_{c_1},\ldots, Z_{c_p}$ on every added ancillary qubit.
We have to verify that the codes $\calG$ and $\calG'$ have the same parameters $k$ and $d$.
It follows from the following lemma.
\begin{lemma}
\label{lemma:break}
Let $\calG$ be any subsystem $[n,k,d]$ code. Let
$q$ be one of the physical qubits of $\calG$ and $a$ be an extra ancillary qubit.
Define a new subsystem code
\be
\calG'=\la \calG, X_q X_a,Z_a\ra,
\ee
where all elements of $\calG$ act trivially on the ancillary qubit.
Then $\calG'$ has parameters $[n+1,k,d]$.
\end{lemma}
Indeed, applying the lemma with  $q=c$ and $a=c_1$ we obtain a new code  $\calG'$
with the same parameters $k$ and $d$, with
an extra short-range generator $X_c X_{c_1}$,
and an extra one-qubit generator $Z_{c_1}$.
Although $\calG'$
inherits the long-range generator $X_c X_{c'}\in \calG$, we can now  replace this generator
by $X_{c_1} X_{c'}=(X_c X_{c_1})\cdot (X_c X_{c'})$.
Note that $X_{c_1} X_{c'}$ has a shorter length.
Now we can apply the lemma
again with $q=c_1$ and $a=c_2$
which breaks the generator $X_{c_1} X_{c'}$ into a pair $X_{c_1} X_{c_2}$, $X_{c_2}X_{c'}$
and an extra generator $Z_{c_2}$.
 We can continue this process
until all cells $c_1,\ldots,c_p$ are occupied by ancillary qubits and
the long-range generator $X_c X_{c'}$ is replaced by a chain of short-range
generators, see Eq.~(\ref{chain}), and one-qubit generators $Z_{c_1},\ldots,Z_{c_p}$.

Similar arguments can be applied to
break up all long-range $Z$-type generators. The resulting code $\calG''$ has at most two qubits
at every cell of $A$. It has two-qubit generators $X_c X_{c'}$, $Z_c Z_{c'}$ coupling only horizontal and vertical nearest
neighbor cells respectively. In addition, $\calG''$ has one-qubit generators $Z_c$ and $X_c$ for some cells $c$.
Lemma~\ref{lemma:break} implies that the code $\calG''$
has the same parameters $k$ and $d$.

Now we can easily prove Theorem~\ref{thm:main}. Let $m$ be any sufficiently large integer.
Theorem~\ref{thm:GV} implies that for any $k\le \alpha m$
there exists a matrix $A\in \calB(m,k)$ such that $\dcol(A)\ge \beta m$ and $\drow(A)\ge \beta m$.
By Theorem~\ref{thm:GBS},  the corresponding subsystem code $\calG$ has parameters $[n,k,d]$,
where $n\le m^2$ and $d=\min{(\dcol,\drow)}\ge \beta m$, see Theorem~\ref{thm:GBS}.
Transforming $\calG$ into a local form as explained above increases
the number of physical qubits due to the addition of ancillas. Since each cell
of $A$ contains at most two qubits, we can get a  $[2m^2,k,d]$ code.
(Some cells of $A$ may contain only one qubit or no qubits at all.
We can add extra unused gauge qubits to each of those cells thus making the total number of qubits $2m^2$.)
In the rest of this section we prove Lemma~\ref{lemma:break}.
\begin{proof}
Consider an auxiliary subsystem code with a gauge group $\tilde{\calG}=\la \calG, X_a, Z_a\ra$,
where all elements of $\calG$ act trivially on the ancillary qubit.
Clearly $\tilde{\calG}$ is obtained from $\calG$ by adding one gauge qubit and thus
$\tilde{\calG}$ has parameters $[n+1,k,d]$.
Let $U$ be the CNOT gate with the control qubit $a$ and a target qubit $q$,
see Fig.~\ref{fig:cnot}.
One can easily check  that $U$ maps $\tilde{\calG}$ to $\calG'$, that is,
\be
\label{U}
\calG'=\{ P'=U\tilde{P}U^\dag, \quad \tilde{P}\in \tilde{\calG}\}.
\ee
\begin{figure}[htb]
\centerline{
\mbox{
 \includegraphics[height=2cm]{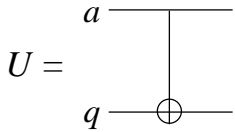}
 }}
\caption{The conjugation by $U$ maps generators of $\tilde{G}$ to generators of $\calG'$
since $UX_a U^\dag = X_q X_a$ and $UZ_a U^\dag = Z_a$.}
 \label{fig:cnot}
\end{figure}
It follows that $\calG'$ and $\tilde{\calG}$ have the same number of logical qubits, that is,
$\calG'$ has parameters $[n+1,k,d']$ for some distance $d'$.
Let us show that $d'=d$.
Indeed, Eq.~(\ref{U}) implies that $\tilde{P}$ is a dressed logical operator for $\tilde{\calG}$ iff $P'\equiv U \tilde{P} U^\dag$ is
a dressed logical operator for $\calG'$. Let $\tilde{P}$ be a dressed logical operator of $\tilde{\calG}$
with the minimum weight, that is, $|\tilde{P}|=d$.
Since $X_a,Z_a\in \tilde{\calG}$, minimality of the weight implies that
$\tilde{P}$ acts trivially on the qubit $a$.
Then $P'=U\tilde{P}U^\dag$ is a dressed logical operator for $\calG'$ and $P'$ acts on the qubit $a$
either by identity or by $Z$. More specifically, $P'=\tilde{P} Z_a^\alpha$ where $\alpha=0$ iff
$\tilde{P}$ acts on $q$ by $I$ or $X$, and $\alpha=1$ iff $\tilde{P}$ acts on $q$ by $Z$ or $Y$.
Since $Z_a\in \calG'$ we conclude that $\tilde{P}$ is also a dressed logical operator for $\calG'$
and thus $d'\le |\tilde{P}|=d$.

Conversely, let $P'$ be a dressed logical operator of $\calG'$ with the minimum weight, that is, $|P'|=d'$.
Using the gauge operators $X_q X_a$ and $Z_a$ we can cancel the action of $P'$ on the qubit $a$
without increasing its total weight, so we can additionally assume that $P'$ acts trivially on $a$.
Then $\tilde{P}=U^\dag P' U$ is a dressed logical operator for $\tilde{\calG}$ and
$\tilde{P}=P' Z_a^\alpha$ for some $\alpha\in \{0,1\}$.
Since $Z_a\in \tilde{\calG}$ we conclude that $P'$ is also a dressed logical operator for $\tilde{\calG}$ and
thus $d\le |P'| =d'$. We have shown that $d'=d$.
\end{proof}

\section{The upper bound: technical tools}
\label{sec:tech}
Let us begin by introducing some more notations.
We shall assume that the $n$ physical qubits occupy sites of a 2D lattice $\Lambda$
of size $\sqrt{n}\times \sqrt{n}$.
For any subgroup $\calG\subseteq \calP$ and any subset of qubits $M\subseteq \Lambda$
introduce a group
\[
\calG(M)=\{ P\in \calG\, : \, \mathrm{Supp}(P)\subseteq M\}
\]
which includes all elements of $\calG$ whose support is contained in
$M$. In particular, $\calP(M)$ is a group of all Pauli operators
whose support is contained in $M$.
For any subset $M\subseteq \Lambda$ let $\overline{M}=\Lambda\backslash M$
be the complement of $M$.
Introduce also a group
\[
\calG_M=\{ P\in \calP(M)\, : \, PQ\in \calG \quad \mbox{for
some $Q\in \calP(\overline{M})$}\}
\]
which includes all Pauli operators $P\in \calP(M)$  that can be
extended to some element of $\calG$.
In other words  $\calG_M$ is a group obtained by
restricting elements in $\calG$ to $M$.
By definition $\calG(M)\subseteq \calG_M\subseteq
\calP(M)\subseteq \calP$.

Throughout this paper we shall ignore overall phase factors of Pauli operators.
Then the Pauli group $\calP$ can be regarded as the $2n$-dimensional binary space~\cite{CRSS97}
such that  multiplication of the Pauli operators corresponds to addition of the binary
strings modulo two. For any subgroup $\calG\subseteq \calP$ define its dimension $\dim{\calG}$
as the smallest number of Pauli operators generating $\calG$. In particular, $\dim{\calP}=2n$
and $\dim{\calP(M)}=2|M|$.
The standard stabilizer formalism~\cite{CRSS97} implies that
\[
\calC(\calC(\calG))=\calG
\]
and
\[
\dim{\calC(\calG)}+\dim{\calG}=2n
\]
for any subgroup $\calG$.
We shall need the following simple fact.
\begin{fact}
\label{fact:1}
Let $\calG\subseteq \calP$ be any subgroup and $M\subseteq \Lambda$ be any subset of qubits.
Suppose $2|M|<\dim{\calG}$. Then
$\calG$ contains at least one non-identity operator acting trivially on $M$.
\end{fact}
\begin{proof}
Indeed, the constraint that $P$ acts trivially on $M$ leads to a system
$2|M|$ binary linear equations which has a non-trivial solution whenever  $2|M|<\dim{\calG}$.
\end{proof}

Consider a subsystem code with a gauge group $\calG$
and a stabilizer group $\calS=\calG\cap \calC(\calG)$.
For any subset of qubits $M\subseteq \Lambda$ let
$l(M)$ be the number of independent dressed logical operators
supported inside $M$, that is,
\[
l(M)=\dim \calC(\calS_M) \cap \calP(M) -\dim \calG(M).
\]
Let $l_{\bare}(M)$ be the number of independent bare logical operators
supported inside $M$,
\[
l_{\bare}(M)=\dim \calC(\calG_M) \cap \calP(M) - \dim \calS(M).
\]
Our first technical tool will be the following lemma which was originally proved for
stabilizer codes by Yoshida and Chuang~\cite{Yoshida10}.
\begin{lemma}
\label{lemma:cleaning}
Suppose a subsystem code $\calG$ has $k$ logical qubits. Then
\[
l_{\bare}(M)+l(\overline{M})=2k
\]
for any subset of qubits $M\subseteq \Lambda$.
\end{lemma}
It can also be regarded as a generalization of the ``Cleaning Lemma" proved for subsystem codes in~\cite{BT08}.
The Cleaning Lemma of~\cite{BT08} dealt only with a special case  $l_{\bare}(M)=0$ or $l(M)=0$.
Specializing Lemma~\ref{lemma:cleaning} to stabilizer codes ($\calG=\calS$) one gets a simpler
statement $l(M)+l(\overline{M})=2k$ which coincides with the result proved in~\cite{Yoshida10}.
Note that Lemma~\ref{lemma:cleaning} does not need any spatial locality properties.
\begin{proof}
Let $m\equiv l_{\bare}(M)$ and let $P_1,\ldots,P_m$
be $m$ independent bare logical operators supported inside $M$, that is,
\be
\label{P_a}
P_a\in \calC(\calG)\backslash \calG \quad \mbox{and} \quad \mathrm{Supp}(P_a)\subseteq M
\ee
for all $a=1,\ldots,m$. Since the code has $k$ logical qubits,
one can choose $t=2k-m$ independent bare logical operators $Q_1,\ldots,Q_t \in \calC(\calG)\backslash \calG$
commuting with $P_1,\ldots,P_m$. Then we have
\be
\label{centr1}
\calC(\calG_M)\cap \calP(M)=\la \calS(M),P_1,\ldots,P_m\ra.
\ee
and
\be
\label{PQ=QP}
P_a Q_b=Q_b P_a, \quad \forall \; a,b.
\ee
Taking the centralizer of both parts of Eq.~(\ref{centr1}) one arrives at
\be
\label{G_M}
\calG_M=\calP(M)\cap \calC(\calS(M)) \cap \calC(P_1) \cap \ldots \cap \calC(P_m).
\ee
Represent $Q_a=Q_a^{\rm in} Q_a^{\rm out}$, where $Q_a^{\rm in}$ and $Q_a^{\rm out}$ are the restrictions
of $Q_a$ onto $M$ and $\overline{M}$ respectively. Combining Eqs.~(\ref{P_a},\ref{PQ=QP},\ref{G_M})
and taking into account that $P_1,\ldots,P_m$ have support only on $M$
one gets
\[
Q_a^{\rm in} \in \calG_M \quad \mbox{for all $a=1,\ldots,t$}.
\]
By definition of $\calG_M$, there must exist gauge operators $G_a\in \calG$ extending $Q_a^{\rm in}$, that is,
$Q_a'\equiv  Q_a G_a \in \calP(\overline{M})$ for all  $a=1,\ldots,t$.
It follows that  $Q_1',\ldots,Q_t'$ are independent dressed logical operators
supported on $\overline{M}$, that is, $l(\overline{M})\ge t$. On the other hand, any dressed logical operator
supported on $\overline{M}$ must commute with $P_1,\ldots,P_m$, and thus $l(\overline{M})\le t$.
We have proved that $l(\overline{M})=t$.
\end{proof}
Our second technical tool is the ``Restriction Lemma" proved in~\cite{BT08}
which relates the distance of a subsystem code $\calG$ defined on the entire lattice $\Lambda$
to the distance of a code $\calG_M$ obtained by restricting $\calG$
onto some subset of qubits $M\subseteq \Lambda$.
Note a generators of $\calG_M$ can be chosen as generators of $\calG$ restricted to $M$.
Assuming that $\calG$ has spatially local generators with some interaction range $r$, the same
is true for $\calG_M$.
\begin{dfn}
Given an interaction range $r$ and a subset $M\subseteq \Lambda$
let $\partial M\subseteq \overline{M}$ be the set of all sites in $\overline{M}$
that lie
within distance $r$ from $M$.
\end{dfn}
\begin{lemma}[Restriction Lemma]
\label{lemma:restrict}
Suppose a gauge group $\calG$ has spatially local generators
with an interaction range $r$.
Choose any subset $M\subseteq \Lambda$  and consider the
subsystem code with the gauge group $\calG_M$.
Then one of the following is true:\\
(1) The code $\calG_M$ has no logical qubits,\\
(2) The code $\calG_M$ has distance at least $d-|\partial M|$\\
where $d$ is the distance of $\calG$.
\end{lemma}

\section{Holographic principle for error correction}
\label{sec:holo}

Let $\calG$ be the gauge group of some subsystem $[n,k,d]$ code. We shall fix some choice
of spatially local generators of $\calG$ such that each generator
has support in a box of size $r\times r$ for some interaction range $r=O(1)$.
By definition of the distance $d$, no subset $M\subseteq \Lambda$ of less than $d$
qubits can support a dressed logical operator, that is, $l(M)=0$ whenever $|M|<d$.
Accordingly, one can choose $R$  proportional to $\sqrt{d}$ such that no square box of size $R\times R$ supports
 a dressed  logical operator.
In this section we derive a much stronger condition
which  can be regarded as an analogue of the famous holographic principle.
\begin{lemma}
\label{lemma:Hmain}
One can choose $R=\Omega(d)$ such that no square box of size $R\times R$ supports a dressed  logical operator.
\end{lemma}
Loosely speaking, the lemma asserts that a non-trivial dressed logical operator cannot be supported on a region
whose {\em perimeter} is smaller than the distance $d$ (with some constant coefficient depending on $r$)
even if the number of qubits in the interior of the region is much larger than $d$.
A similar result was obtained in~\cite{Bravyi09} for 2D stabilizer codes.
We shall begin by proving an auxiliary lemma.
\begin{lemma}
\label{lemma:Haux}
Let $A,B$ be any disjoint subsets of qubits such that $l(A)=0$
and $|B|+|\partial \overline{A}|<d$. Then $l(AB)=0$.
\end{lemma}
\begin{proof}
Let $C$ be the complement of $AB$ such that $\Lambda=ABC$.
Applying Lemma~\ref{lemma:cleaning} to the subset $A$ we conclude that
$l_{\bare}(BC)=2k$. Hence we can choose a complete set of
$2k$ bare logical operators  $\overline{X}_1,\overline{Z}_1,\ldots,\overline{X}_k,\overline{Z}_k\in \calC(\calG)\backslash \calG$
supported inside $BC$. Consider a subsystem code specified by the gauge group $\calG_{BC}$
that involves only qubits of $BC$.
We claim that $\calG_{BC}$ has $k$ logical qubits and $\overline{X}_a,\overline{Z}_a$ are the bare logical operators
of $\calG_{BC}$. Indeed, since $\overline{X}_a,\overline{Z}_a$ have support on $BC$ and commute with $\calG$
we infer that $\overline{X}_a,\overline{Z}_a\in \calC(\calG_{BC})$ for all $a$. Accordingly,  $\overline{X}_a,\overline{Z}_a\notin \calG_{BC}$
since $\overline{X}_a\overline{Z}_a=-\overline{Z}_a\overline{X}_a$ for all $a$. It shows that $\calG_{BC}$ has at least $k$ logical qubits.
Conversely, let $\overline{P}\in \calC(\calG_{BC})\backslash \calG_{BC}$ be any bare logical operator of $\calG_{BC}$.
By definition of $\calG_{BC}$ it implies that $\overline{P}\in \calC(\calG)$ and $\overline{P}\notin \calG$, that is,
$\overline{P}$ must be a bare logical operator of the original code $\calG$ in which case it can be expressed in terms
of $\overline{X}_a,\overline{Z}_a$. We have shown that $\calG_{BC}$ has $k$ logical qubits and $\overline{X}_a,\overline{Z}_a$ are the bare logical operators of $\calG_{BC}$.
Applying the Restriction Lemma to the code $\calG_{BC}$ we conclude that $\calG_{BC}$
has distance $d'\ge d-|\partial (BC)|=d-|\partial \overline{A}|$. The assumptions of the lemma then imply $d'>|B|$.
Thus the code $\calG_{BC}$ has no dressed logical operators supported inside $B$, that is, $l'(B)=0$
(here and in the rest of the proof all quantities labeled by a prime refer to the code $\calG_{BC}$).
Applying Lemma~\ref{lemma:cleaning} to the code $\calG_{BC}$ and the subset $B$ we infer that
$l'_{\bare}(C)=2k$. Hence we can choose a complete set of bare logical operators
$\overline{X}_a',\overline{Z}_a'\in \calC(\calG_{BC})\backslash \calG_{BC}$ supported on $C$.
But then $\overline{X}_a',\overline{Z}_a'$ are also bare logical operators of the original code $\calG$
which implies $l_{\bare}(C)=2k$ for the code $\calG$. Applying Lemma~\ref{lemma:cleaning}
to the code $\calG$ and the subset $C$
we arrive at $l(AB)=0$.
\end{proof}
Now we are ready to prove Lemma~\ref{lemma:Hmain}.
\begin{proof}
Choose any $R$ such that $rR\ll d$.
For any square box $M$ of size $R\times R$
consider a sequence of square boxes
$A_1\subset A_2 \subset \ldots \subset A_p=M$
such that $A_1$ has cardinality $|A_1|<d$ and $A_{i+1}$ is the smallest box that contains $A_i$ and
the boundary of $A_i$.  Let $B_i=A_{i+1}\backslash A_i$. Then
$|B_i|+|\partial \overline{A_i}|\le O(1)rR<d$ for all $i=1,\ldots,p$. Since $|A_1|<d$ we have $l(A_1)=0$.
Applying Lemma~\ref{lemma:Haux} inductively with $A\equiv A_i$ and $B\equiv B_i$ we arrive at
$l(A_p)=l(M)=0$.
\end{proof}

\section{Proof of the upper bound}
\label{sec:proof}

Now we are ready to prove the upper bound Eq.~(\ref{bound2}).
By assumption,  the support of any generator of the gauge group $\calG$ can be
covered by a square block of size $r\times r$ for some interaction range $r=O(1)$.
Consider a partition of the lattice $\Lambda=AB$ shown on Fig.~\ref{fig:weak}.
The region $A$ consists of square blocks $A_1,\ldots, A_m$ of size $R\times R$
with $R=\Omega(d)$ such that $l(A_i)=0$, see Lemma~\ref{lemma:Hmain}.
We choose the separation between adjacent blocks in $A$ at least $r$ such that
any generator of $\calG$ overlaps with at most block in $A$.
We claim that
\be
\label{Abare}
l_{\bare}(A)=0.
\ee
Indeed, suppose $P\in \calC(\calG)\backslash \calG$ is a bare logical operator supported on $A$.
Let $P_i$ be the restriction of $P$ onto a block $A_i$ such that $P=P_1P_2\cdots P_m$.
Since any generator of $\calG$ overlaps with at most one block in $A$, we have
$P_i\in \calC(\calG)$ for all $i$. However there must exist at least one block $A_i$ such that
$P_i\notin \calG$ since otherwise $P\in \calG$. Then for such a block we have
$P_i\in \calC(\calG)\backslash \calG$, that is, $P_i$ is a bare logical operator supported inside $A_i$.
But this implies $l(A_i)\ge l_{\bare}(A_i)>0$ which is a contradiction. It proves Eq.~(\ref{Abare}).
Applying Lemma~\ref{lemma:cleaning} we get
\[
l(B)=2k.
\]
However, a subset $B$ can support at most $2|B|$ independent Pauli operators
which implies $2|B|\ge l(B)$, that is, $|B|\ge k$. Simple algebra shows that
$|B|=O(n/R)=O(n/d)$ and thus $kd=O(n)$.

Let us now prove the stronger bound Eq.~(\ref{bound1}) assuming that both $\calG$ and $\calS$ have
spatially local generators with a constant interaction range $r$ and $r_s$ respectively.
Consider a partition of the lattice $\Lambda=ABC$ shown on Fig.~\ref{fig:strong}.
The regions $A$, $B$ consist of blocks $A_1,\ldots, A_m$
and $B_1,\ldots,B_m$ respectively
of size  $R\times R$ with $R=\Omega(d)$ such that $l(A_i)=0$ and $l(B_i)=0$,
see Lemma~\ref{lemma:Hmain}.
The region $C$ consists of disks of radius $\max{\{r,r_s\}}$ so that
adjacent blocks in $A$ and adjacent blocks in $B$
are separated from each other by distance $\max{\{r,r_s\}}$.
Then we can choose generators in $\calG$ and $\calS$ such that any generator
overlaps with at most one block in $A$ and with at most one block in $B$.
Applying the same arguments as above we get $l_{\bare}(A)=0$ and thus
Lemma~\ref{lemma:cleaning} implies
\[
l(BC)=2k.
\]
Let us assume that
\be
\label{assume}
|C|<k
\ee
and show that it leads to a contradiction.
Indeed, choose any set of $2k$ independent dressed logical operators $P_1,\ldots,P_{2k}\in \calC(\calS)\backslash \calG$
supported inside $BC$ and let $\calQ=\la P_1,\ldots,P_{2k}\ra$. Applying Fact~\ref{fact:1}
to region $C$ and the group $\calQ$ we conclude that
there exists at least one non-trivial dressed logical operator $P\in \calC(\calS)\backslash \calG$ supported only inside $B$.
Let $P_i$ be the restriction of $P$ onto a block $B_i$
such that $P=P_1P_2\cdots P_m$. Since any generator of $\calS$
overlaps with at most one block in $B$ we conclude that $P_i \in \calC(\calS)$. However,
there must exist at least one block $B_i$ such that $P_i\notin \calG$ since otherwise $P\in \calG$.
Then $P_i$ is a non-trivial dressed logical operator, that is, $l(B_i)>0$ which is a contradiction.
Hence Eq.~(\ref{assume}) is impossible and we have $|C|\ge k$.
Simple algebra shows that $|C|=O(n/R^2)=O(n/d^2)$ which yields $kd^2=O(n)$.

\begin{figure}[htb]
\centerline{\includegraphics[height=3cm]{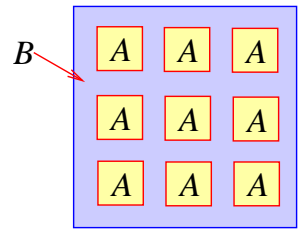}}
\caption{The partition of the lattice $\Lambda=AB$ used to prove the bound
$kd=O(n)$. The region $A$ consists of square blocks $A_1,\ldots, A_m$ of size $R\times R$
with $R=\Omega(d)$ such that $l(A_i)=0$. Adjacent blocks are separated from each other by distance $r=O(1)$.
It implies $l_{\bare}(A)=0$ and thus $l(B)=2k$. This is possible only if $|B|\ge k$ which yields
$kd=O(n)$.}
\label{fig:weak}
\end{figure}

\begin{figure}[htb]
\centerline{\includegraphics[height=3cm]{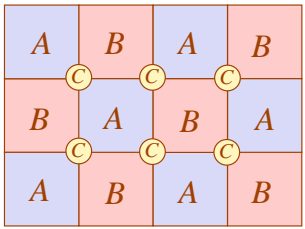}}
\caption{The partition of the lattice $\Lambda=ABC$ used to prove the bound
$kd^2=O(n)$. The regions $A$, $B$ consist of blocks $A_1,\ldots, A_m$
and $B_1,\ldots,B_m$ respectively
of size  $R\times R$ with $R=\Omega(d)$ such that $l(A_i)=0$ and $l(B_i)=0$.
The region $C$ consists of disks of radius $\max{\{r,r_s\}}$ so that
adjacent blocks in $A$ and adjacent blocks in $B$
are separated from each other by distance $\max{\{r,r_s\}}$.
It implies $l_{\bare}(A)=0$ and thus $l(BC)=2k$. If $k$ violates the upper bound one
would have $l(B)\ge l(BC)-2|C|>0$. Assuming that $\calS$ has spatially local generators
this is possible only if $l(B_i)>0$ for some block $B_i$ which is a contradiction.
 }
\label{fig:strong}
\end{figure}

\section{Conclusions and open problems}
\label{sec:concl}
In this paper we have studied subsystem codes for which the gauge group has spatially
local generators in the 2D geometry. It was shown that the parameters $[n,k,d]$ of such codes
must  obey an upper bound
$kd=O(n)$. We have also introduced a family of codes, the generalized Bacon-Shor codes, that achieves this bound with
both $k$ and $d$ proportional to $\sqrt{n}$.  The gauge group of the generalized Bacon-Shor codes involves only two-qubit
generators of type $XX$ and $ZZ$ coupling nearest-neighbor qubits (and some one-qubit generators).
It follows that the syndrome measurement for these codes requires only eigenvalue measurements for operators
$XX$ and $ZZ$ on nearest-neighbor qubits.  Our proof of existence presented in Sections~\ref{sec:GBS},\ref{sec:GV}
is not constructive since it requires binary matrices  achieving the Gilbert-Varshamov bound stated in Theorem~\ref{thm:GV}. On the other hand, one can easily show that a random $m\times m$ binary matrix $A$ with a fixed rank $k$ achieves the Gilbert-Varshamov bound of Theorem~\ref{thm:GV} with probability approaching one in the limit $m\to \infty$.
Therefore for finite sufficiently small lattice sizes 
one can  simply choose the desired  binary matrix $A$ randomly (with a fixed rank), compute the minimum distances
$\dcol$, $\drow$ and check whether the Gilbert-Varshamov bound is satisfied. 

A serious drawback of the standard 2D Bacon-Shor code~\cite{Bacon06} that  precludes it from
being used in the topological quantum computation schemes is the lack of 
a constant error threshold in the limit of large lattice size~\cite{Cross07,Pastawsky09}. We expect that the same
drawback is shared by the generalized Bacon-Shor codes introduced in the paper. 
It is therefore an interesting open problem whether it is possible to construct 2D subsystem codes
with both $k$ and $d$ proportional to $\sqrt{n}$ which would have a good behavior under random uncorrelated errors.

Finally, let us point out that our construction of the generalized Bacon-Shor codes naturally extends to
3D Bacon-Shor codes~\cite{Bacon06}. In the 3D case the binary matrix $A$ should be replaced by a
three-dimensional binary array with qubits occupying cells with $A_{i,j,k}=1$. The corresponding gauge
group $\calG$ is generated by operators $XX$, $YY$, and $ZZ$ coupling pairs of qubits that differ only in
$x$, $y$, and $z$-coordinate respectively. For any choice of the array $A$ the resulting subsystem code can be transformed
into the spatially local form by introducing ancillary qubits and simulating every long-range generator by a chain of short-range 
generators as described in Section~\ref{sec:local}. Finding the optimal scaling of $d$ and $k$ for such
generalized 3D Bacon-Shor codes is an interesting open problem.

\section*{Acknowledgments}
The author would like to thank Graeme Smith and Barbara Terhal for useful discussions.
This work was partially supported by  DARPA QUEST program under contract number HR0011-09-C-0047.

\section*{Appendix A}
The purpose of this section is to prove the upper bound Eq.~(\ref{GBSbound1+}).
We shall also construct a family of codes that achieves this bound.

We will  show that  a tuple $[n,k,\drow,\dcol]$ can be realized by some binary matrix $A$
only if  the following quadratic optimization problem has feasible solutions:
\bea
r_x & \ge 0 & \quad \quad  \forall x\in \Sigma^k \label{opt1} \\
c_x& \ge 0 & \quad \quad \forall x\in \Sigma^k  \label{opt2} \\
\sum_{x\, : \, x\cdot y=1}\; r_x &\ge \drow&  \quad \quad \forall y\in \Sigma^k\backslash 0,  \label{opt3} \\
\sum_{x\, : \, x\cdot y=1}\; c_x &\ge \dcol&  \quad \quad \forall y\in \Sigma^k\backslash 0, \label{opt4} \\
\sum_{x,y\, : \, x\cdot y=1}\; r_x c_y &=n. \label{opt5}&
\eea
Here $r_x$ and $c_x$ are integer-valued variables labeled by binary strings $x\in \Sigma^k\equiv \{0,1\}^k$.
We used the notation $x\cdot y\equiv \sum_{i=1}^k x_i y_i \pmod{2}$ for the binary inner product.
Hence we can get a lower bound on $n$ by minimizing the quadratic function of $r_x,c_y$ defined in Eq.~(\ref{opt5})
subject to constraints Eqs.~(\ref{opt1}-\ref{opt4}).

Let us begin by deriving analogous optimization problem corresponding to ordinary classical codes.
Let $G$ be the generating matrix of some classical $[n,k,d]$ code, such that $G$ has size $k\times n$
and the rows of $G$ form the basis of the codespace. For any binary string $x=[x_1,\ldots,x_k]\in \Sigma^k$ let
$n_x$ be the number of columns $[x_1,\ldots,x_k]^T$ in the matrix $G$.
Any codeword can be represented as $y G$ for some binary string $y\in \Sigma^k$. One can easily check that
the Hamming weight of $y G$ can be expressed as
\[
|y G|=\sum_{x\, : \, x\cdot y=1}\, n_x,
\]
where the summation is over binary strings $x\in \Sigma^k$.
Hence a tuple $[n,k,d]$ can be realized by some code iff the following optimization problem has feasible solutions:
\bea
n_x & \ge 0 & \quad \quad  \forall x\in \Sigma^k \label{opt1c} \\
\sum_{x\, : \, x\cdot y=1}\; n_x &\ge d&  \quad \quad \forall y\in \Sigma^k\backslash 0,  \label{opt2c} \\
\sum_{x\in\Sigma^k}\, n_x=n. \label{opt3c}
\eea
Here we treat $n_x$ as integer-valued variables. Indeed, any solution $\{n_x\}$ can be transformed
into a generating matrix $G$ (defined uniquely up to permutation of columns) of size $k\times n$.
Then Eq.~(\ref{opt3c}) implies that $G$ represents a $[n,k,d]$ classical code.

Now consider an arbitrary binary matrix $A$ of rank $k$. Without loss of generality
the first $k$ rows and the first $k$ columns of $A$ are linearly independent (otherwise
permute rows or columns). Let $\Grow$ be the generating matrix of a classical
code spanned by the first $k$ rows of $A$. Similarly, let $\Gcol$ be the generating matrix of a classical
code spanned by the first $k$ rows of $A^T$. (Note that $\Grow$ and $\Gcol$ may have different
length if $A$ is not a square matrix.) Let $c_x$ be the number of columns $x=[x_1,\ldots,x_k]^T$ in $\Gcol$.
Let $r_x$ be the number of columns  $x=[x_1,\ldots,x_k]^T$ in $\Grow$.
The variables $c_x,r_x$ must obey inequalities analogous to Eq.~(\ref{opt1c},\ref{opt2c}) since we assumed
that $\Gcol$ and $\Grow$ have distance $\dcol$ and $\drow$ respectively.
It yields Eqs.~(\ref{opt1}-\ref{opt4}). It remains to derive Eq.~(\ref{opt5}).
Consider any row of $A$ that starts with $x=[x_1,\ldots,x_k]$. It can be represented
as $z \Grow$ for some $z=z(x)\in \Sigma^k$ since by assumption any row of $A$ is a linear combination
of the first $k$ rows. Moreover, the function $z(x)$ must be linear and invertible, that is,
$z(x)=xM$ for some $k\times k$ invertible matrix $M$.
As before, the Hamming weight of $z\Grow$ can be expressed as
\[
|z\Grow|=\sum_{y\, : \, y\cdot z=1}\, r_y,
\]
where the summation is over binary strings $y\in \Sigma^k$.
Since the number of rows in $A$ that start from $x=[x_1,\ldots,x_k]$ is equal to $c_x$ we arrive at
\[
n=|A|=\sum_{x\in \Sigma^k} c_x \, |z(x) \Grow|=\sum_{x,y\, : \, xM\cdot y=1}\, c_x r_y.
\]
Here $xM\cdot y$ is the inner product between binary strings $xM$ and $y$.
Since Eqs.~(\ref{opt1}-\ref{opt4}) are invariant under a change of variables
$c_x\to c_{xM}$ for any invertible matrix $M$, we get Eq.~(\ref{opt5}).

Now we can easily prove Eq.~(\ref{GBSbound1+}).
Let $r=\sum_{x\ne 0} r_x$. Adding up Eq.~(\ref{opt3}) for all $y\ne 0$ we
count each $r_x$ exactly $2^{k-1}$ times, that is, we get
\[
r\ge \drow (2^k-1)2^{1-k}=\drow(2-2^{1-k}).
\]
Then combining  Eqs.~(\ref{opt4},\ref{opt5}) we get
$n\ge \dcol r$ which is equivalent to Eq.~(\ref{GBSbound1+}).

Let us show that Eq.~(\ref{GBSbound1+}) is tight. Choose any integer $k\ge 1$ and define a  matrix
$A$ of size $(2^k-1)\times (2^k-1)$ as a binary version of the Hadamard matrix,
\[
A_{x,y}=x\cdot y
\]
with $x,y\in \Sigma^k\backslash 0$. Obviously, $A$ is a symmetric matrix. Its row-space and its
column-space coincide with the Hamming code $[2^k-1,k,2^{k-1}]$.
In particular,
any row and any column of $A$ have weight $2^{k-1}$. Thus we get
$\drow=\dcol=2^{k-1}$ and $n=(2^k-1)2^{k-1}$. It achieves the bound Eq.~(\ref{GBSbound1+}).

\bibliographystyle{hunsrt}

\end{document}